\documentclass{article}
\usepackage{cite}
\usepackage{amsmath}
\usepackage[utf8]{inputenc}
\usepackage[english]{babel}
\usepackage{graphicx}
\usepackage{amsfonts}
\usepackage{bbm}
\usepackage{amsthm}
\usepackage{amssymb}
\usepackage{xcolor}
\usepackage{float}
\usepackage{url}
\usepackage{csquotes}
\usepackage{gensymb}
\usepackage{caption}
\usepackage{subcaption}
\usepackage{accents}
\usepackage{soul,xcolor}

\usepackage{optidef}
\usepackage{epsfig}
\usepackage{longtable}
\usepackage{lineno}

\usepackage{changes}

\usepackage{hyperref}
\usepackage[ruled,vlined]{algorithm2e}
\SetArgSty{textnormal}

\newtheorem{theorem}{Theorem}

\newtheorem{proposition}{Proposition}

\newtheorem{remark}{Remark}
\newtheorem{assumption}{Assumption}

\setstcolor{blue}

\begin{document}

\title{Control Protocol Design and Analysis for Unmanned Aircraft System Traffic Management}
\author{Jiazhen Zhou, Dawei Sun, Inseok Hwang, Dengfeng Sun
\thanks{Jiazhen Zhou, Dawei Sun, Inseok Hwang and Dengfeng Sun are with the School of Aeronautics and Astronautics, Purdue University, West-Lafayette, IN 47906 USA; Emails: \texttt{\{zhou733, sun289, ihwang, dsun\}@purdue.edu}.}}

\maketitle

\begin{abstract}
Due to the rapid development technologies for small unmanned aircraft systems (sUAS), the supply and demand market for sUAS is expanding globally. With the great number of sUAS ready to fly in civilian airspace, an sUAS aircraft traffic management system that can guarantee the safe and efficient operation of sUAS is still at absence. In this paper, we propose a control protocol design and analysis method for sUAS traffic management (UTM) which can safely manage a large number of sUAS. The benefits of our approach are two folds: at the top level, the effort for monitoring sUAS traffic (authorities) and control/planning for each sUAS (operator/pilot) are both greatly reduced under our framework; and at the low level, the behavior of individual sUAS is guaranteed to follow the restrictions. Mathematical proofs and numerical simulations are presented to demonstrate the proposed method. 
\end{abstract}

\section{Introduction}

With the vast investment of financial support and research effort, sUAS are envisioned to achieve autonomy based on the rapid development of technologies including guidance, communication, sensing, and control. Commercial sUAS have been developed for  a variety of tasks, such as package delivery, rescue operations, photography, surveillance, infrastructure monitoring, and etc. The market for sUAS of civilian purposes is expanding rapidly with the potential users including companies, governments, and hobbyists. With expectation of a great number of sUAS operating in the airspace system, especially in the urban environment, the control of sUAS behavior and management of their traffic are crucial for safety and efficiency~\cite{aweiss2018unmanned}. With the above concerns, some rules and laws to regulate the operation of sUAS in the civil domain have been published by the Federal Aviation Administration (FAA)~\cite{federal2016summary}. However, a traffic management system that can ensure the enforcement of the rules and the efficiency of the system is still at absence. 

The need for sUAS Traffic Management (UTM) System has long been recognized with the increasing number of registered sUAS. The FAA and NASA are leading efforts for making the rules and conducting the research on the large scale sUAS operations. A build-a-little-test-a-little strategy is currently used to address the scalability~\cite{johnson2017flight}. The UTM research is divided into four Technology Capability Levels (TCL)~\cite{NASAUTM}:
\begin{itemize}
	\item  achieve rural UAS operations for agriculture, firefighting and infrastructure monitoring.
	\item realize beyond-visual line-of-sight operations in sparsely populated areas, and provide flight procedures and traffic rules for longer-range applications.
	\item include cooperative and uncooperative UAS tracking capabilities to ensure collective safety of manned and unmanned operations over moderately populated areas.
	\item involve UAS operations in higher-density urban areas for tasks such as news gathering and package delivery, and  large-scale contingency mitigation.
\end{itemize}
The flight tests for TCL 1 and TCL 2 have been successfully conducted in NASA's test cites, and basic requirements for sUAS operation in less populated areas have been proposed based on test results~\cite{johnson2017flight,NASAOStest}. The discussion and test on airspace  design, corridors, geofencing, severe weather avoidance, separation management, spacing, and contingency management are the main focus and most challenging parts in TCL 3 and TCL 4 research. However, no results on TCL 3 and TCL 4 have been reported, to the best of our knowledge. 
On the other hand, a few works have conceptually discussed the architecture of UTM system and identified its basic elements~\cite{jiang2016unmanned,ren2017small,ippolito2019autonomy}. These works envision UTM based on the existing Air Traffic Management (ATM) system for crewed aircraft. Nevertheless, such a design may not be feasible for large scale operation or dense traffic in the sense that it requires features like flight authorizations, flight plan review/approval, external data services (weather, intruder), which may suffer from the curse of dimensionality. 

To improve the scalability of UTM, we note some key characteristics differentiating sUAS operation from the existing crewed aircraft operation. For crewed aircraft, the human pilot is capable of directly controlling the behavior of an aircraft. In contrast, the human operator for sUAS with remote control has to rely on the system's autonomous control and/or decision supporting tools to cope with the large scale operation and complicated operational environment~\cite{weiss2006global}. Although this autonomous nature brings more challenges to sUAS hardware/software requirements and risk evaluation, it brings an opportunity for UTM to administrate the sUAS from a control systems prospective. The fact that the behavior of an sUAS is more governed by the autopilot instead of the human operator reveals that UTM may regulate the collective sUAS traffic behavior by adopting certain control protocols: if the sUAS can agree on predefined control protocols in certain airspace, then collective safety/efficiency assurance for sUAS traffic can be converted to a control protocol design problem. By directly regulating the sUAS traffic behavior in the control level, UTM can reduce the effort for trajectory planning and reviewing significantly.

According to the aforementioned idea, 
we propose a control protocol design and analysis method to improve the scalability for UTM. In this framework, we envision that UTM is responsible for publishing control protocols for sUAS operating in each basic traffic element such that the desirable collective traffic behavior is assured without reviewing the high dimensional trajectories of all sUAS explicitly. The basic element of sUAS traffic network considered here is called a single \textit{link}, which is an abstraction of ``road" or ``lane", proposed by NASA~\cite{NASAwhite}. The main ingredient of our framework is based on the artificial potential field (APF) approach to control the behavior of sUAS in each link, which is motivated by successes of the APF approach in  various aerospace applications such as aircraft guidance law design~\cite{lopez1995autonomous}, conflict resolution~\cite{ruchti2014unmanned}, and multi-agent control~\cite{arcak2007passivity}. Upon the agreement of a set of APF functions in the control protocol, sUAS can achieve desired collective behaviors, such as collision avoidance, boundary clearance, and speed regularization. Our framework, on its core, converts the problem of sUAS traffic control to an APF-based decentralized protocol design problem, similar to flocking control. A commonly accepted definition for flocking behavior is given by Reynolds rules~\cite{reynolds1987flocks}: 1) stay close to nearby flockmates, 2) avoid collision with nearby flockmates, and 3) attempt to match the velocity with nearby flockmates. For safe UTM operation, collision avoidance is crucial, and the velocity of each sUAS should conform to the desired/reference speed associated with the link. A representative design of APF-based flocking control has been introduced by Olfati-Saber~\cite{olfati2006flocking}. Following the idea, variants of distributed flocking algorithms have been proposed, i.e., the flocking algorithm under time varying communication network topology~\cite{tanner2007flocking}, the flocking algorithm that considers complex robotics models with non-holonomic constraints~\cite{lei2008flocking,li2009flocking,gouvea2011potential,jin2015consensus,mastellone2008formation}, and the hybrid flocking algorithm for fixed-wing aircraft~\cite{sun2019hybrid}.

It should be noted that our problem of APF-based control protocol design and analysis for UTM is different from any existing flocking control problems. For flocking algorithm design, even though the collective system of interest has a multi-agent nature, the system consists of a fixed group of sUAS. However, in our problem, the system of interest consists of sUAS in a certain traffic link, which is time varying in the sense that some sUAS may enter the link and some sUAS may leave the link at some time instances. From the traffic management perspective, it is desired to investigate the sufficient conditions for the sUAS to enter the traffic link without causing  collision. The answer to such a problem is related to the analysis of the APF-based control protocol design using the Hamiltonian function (which is commonly used as an analogue to the concept of ``energy"~\cite{olfati2006flocking}. Since it is known that the collision avoidance can be guaranteed by limiting the ``energy", the convergence rate for energy is especially important for estimating the capacity and entry rate for safe operation in a link. Indeed, the most challenging part of our theoretic analysis lies in establishing the convergence rate of the proposed control protocol, which has not been discussed in general flocking control problems.

The contributions of this work are: 1) We develop a control protocol design and analysis method which can safely manage sUAS traffic, while improving the scalability of UTM. The problem of collective behavior regularization and safety assurance is formally defined based on control theory; 2) After a formal definition of the sUAS traffic regularization, we design a distributed control protocol for sUAS in a single traffic link. Based on the convergence property of our control algorithm, we propose conditions on sUAS for safely entering a traffic link; and 3) Based on our control protocol, we propose hardware/software requirements on sUAS operating in the large scale traffic system.


The rest of the paper is organized as follows. Section~\ref{design} identifies the roles and responsibilities of each element in the sUAS traffic system in our framework. Section~\ref{fomulation-and-design} formally introduces the formulation of the sUAS traffic regularization in a single traffic link and offers theoretical results. Section~\ref{simulation-section} demonstrates the results via illustrative numerical simulations. Finally, Section~\ref{conclusion} draws conclusions.

\section{Elements in the sUAS Traffic System}\label{design}
We consider a basic network structure of the future large scale sUAS traffic. A network is a fundamental structure of ground traffic and air traffic, and thus the usage of such structure in sUAS traffic has been envisioned by NASA~\cite{NASAUTM}. The traffic network is defined as a set of nodes and links, and each link connects two nodes with specified locations. Each link commits a specified altitude block and corridor width where sUAS can be flown from one location to the other. For each link, the authority may specify the desired speed, top speed, desired separation, and minimum separation for collision avoidance. One way to ensure all requirements are satisfied is to review every filed flight plan and make sure every restriction is satisfied. The flight plan is often a time-position 4D trajectory, and a certain resolution of the trajectory is required to achieve safety assurance. Such high-dimension, high-resolution trajectory checking can be overburdening for UTM, which should be responsible for managing large scale sUAS operations. Another way to efficiently manage the traffic is to assign sUAS control protocols to each individual traffic link, by which the collective safety of sUAS traffic within each link can be guaranteed and restrictions can be satisfied via theoretical analysis.
\begin{figure*}[t!]
	\centering
	\includegraphics[scale = 0.5]{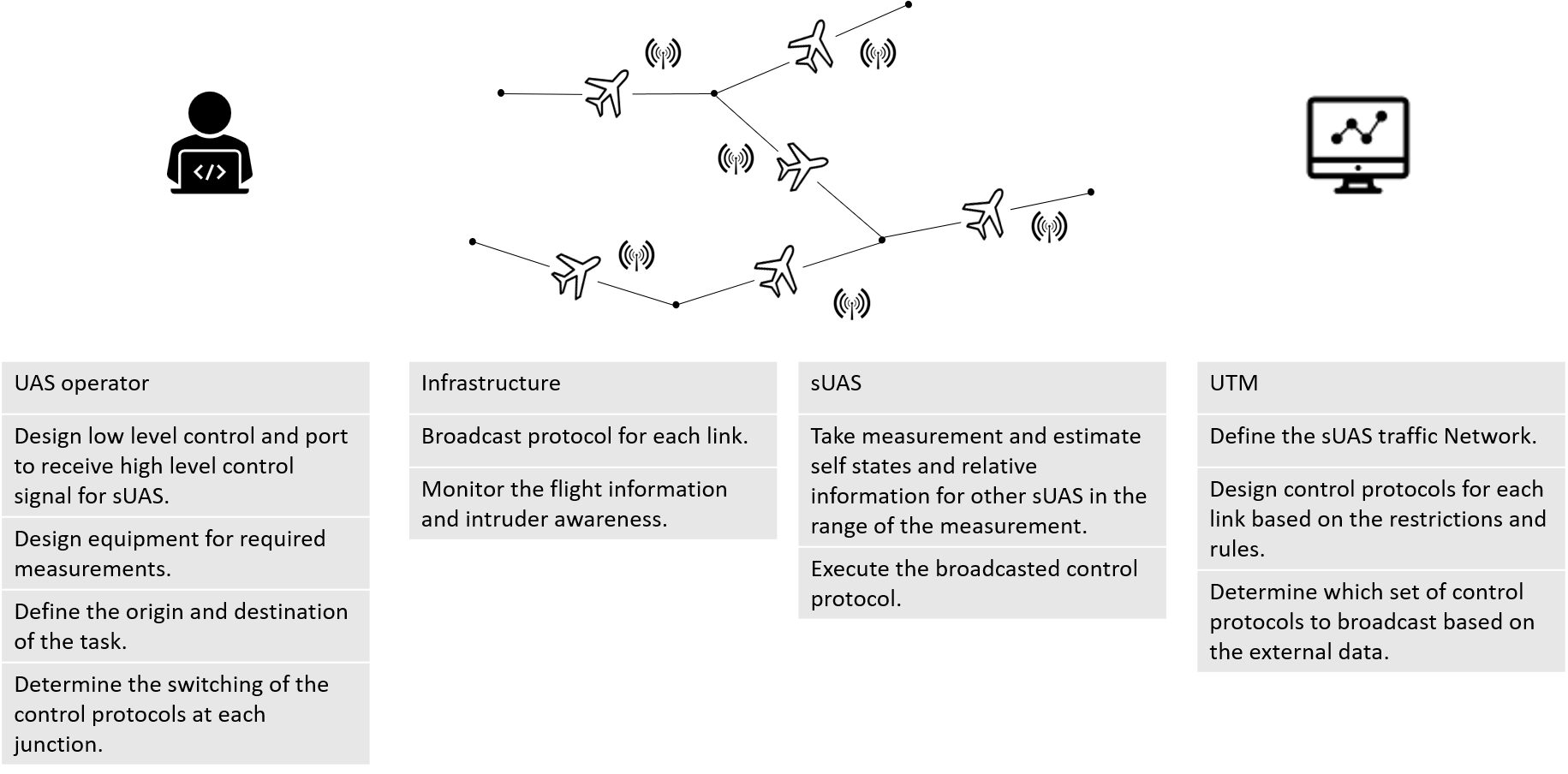}
	\caption{Proposed UTM framework}
	\label{framework}
\end{figure*}

Our framework redefines the roles and responsibilities of four main components in the future sUAS traffic system, which are sUAS operators, infrastructure, sUAS, and UTM. We explain each component in the order of design process:
\begin{enumerate}
	\item UTM design
	\begin{itemize}
	\item Network Definition: A UTM system needs first define the network whose elements include the location of the nodes, and virtual boundary of each link. Each traffic has regulations/rules for sUAS traveling in it, e.g., the speed limit, minimum separation between sUAS, minimum distance to the virtual boundaries, and feasible landing areas. Those regulations may or may not be shared through all links. Separated networks will be necessary for sUAS of different types. For example, fixed-wing sUAS and multi-copters will need different networks due to their distinct flight dynamics, cruising speeds and take-off/landing processes. The construction of the networks needs to be done in collaboration with law makers, such as the FAA.
	\item Control Protocol Design: A set of control protocols can be designed for each link once the rules are defined. In order to have control protocols that are feasible for all sUAS traveling in the network, a basic physical dynamics of sUAS will be assumed. One can assume that the sUAS are equipped with low level autopilot system such that a multi-copter can be viewed as a single integrator model~\cite{kuriki2014consensus}, and a fixed-wing sUAS follows a Dubins car model on the horizontal plane and a double integrator on the vertical direction \cite{sun2019hybrid}. It is desired that under a common control protocol, sUAS can achieve collective safety assurance and operational efficiency. During the operation, the control protocol is broadcast for each link, and sUAS follow the control protocol after entering the link. Based on the control protocol, UTM can propose hardware/software requirements on the on-board measurement/estimation or communication for sUAS.
	\item Control Protocol Selection: Different sets of control protocols need to be designed for each link to take into account different factors, such as weather, human activities, emergencies, etc~\cite{federal2016summary}. During the operations, the best suited control protocols are selected for links and broadcast.
	\end{itemize}
	\item Operator
	\begin{itemize}
		\item Low Level Control Design: The operators design low level controllers, such as from thrust/voltage to accelerations, in order to follow the high level control protocol broadcast by UTM. 
		\item Equipment: Each sUAS should be equipped with necessary sensors and filters to estimate its states and relative information such as the distance to neighboring sUAS or obstacles so that the high level control protocol from UTM can be properly implemented. The sUAS will need to equip with ports that can receive supervisory command from UTM broadcast.
		\item Operation: An operator will need to specify the origin, destination, and the sequence of links the sUAS will travel.
	\end{itemize}
	\item Infrastructure
	\begin{itemize}
		\item Broadcast: It is envisioned that broadcasting will be a good practice for sharing control protocols. By broadcast, control protocols are not hard-coded on each sUAS. The response is immediate if emergency occurs and different control protocols are broadcast.
		\item Monitor: Cameras, LIDAR, and/or radars can also be equipped in the infrastructure to monitor the sUAS traffic and alert for any intruder or malfunction.
	\end{itemize}
	\item sUAS
	\begin{itemize}
		\item State Estimation: During a mission, an sUAS needs to take measurements and estimate the states of itself and relative information with respect to other sUAS or obstacles.
		\item Control Execution: An sUAS receives and executes the control protocol from broadcast.
		\item Control Protocol Switching: When arriving at a junction, multiple control protocols will be available. sUAS select the one specified by the operator before the mission and complete the transition from one link to another.
	\end{itemize}
\end{enumerate}
The overall framework of our design is summarized in Figure~\ref{framework}. It can be seen that in our framework, the operators will only need to determine the sequence of transitions between links. UTM will only need to monitor the real time traffic situations during the daily operation. Collision free, speed limits and other requirements are fulfilled by the design of control protocols. The shift of control design from operators to UTM allows UTM to have a high authority over the behaviors of the sUAS traveling in the network, and thus safety and efficiency can be guaranteed by the collective behavior of sUAS in each link resulting from the common control strategy. 

Each element of the future sUAS traffic system discussed above can be extensively studied. In this article, we limit our attention to one of the most important and challenging elements: the control protocol design and analysis. We develop models and theoretical frameworks for analyzing the sUAS traffic behavior in a single link from a control systems perspective. Our results can be applied to more common and complicated traffic network elements such as merge links and split links~\cite{van2007modeling}. We present our results in details in the following section.

\section{Problem Formulation and Controller Design}\label{fomulation-and-design}
In this section, we formulate the problem of sUAS traffic regularization and present details about how to design and analyze control protocols to regulate the sUAS behavior in a link using artificial potential functions. The objective of this section is to offer a guideline about how to design the APF for each link in an sUAS traffic network and what condition sUAS should satisfy at entry of the link under a certain communication protocol structure such that the speed of all sUAS in the link is regulated and there is no collision or boundary violations.
\subsection{Dynamic model for sUAS}
In this section, we formally introduce the problem of regularization of sUAS traffic in a single link model. Let $\mathcal I$ be the index set of all the sUAS in the traffic system.
We consider a fixed-wing sUAS whose kinematics is described as:
\begin{equation}\label{dynamics}
\begin{aligned}	
\dot x_i &= v_i\cos\theta_i\\
\dot y_i &= v_i\sin\theta_i\\
\dot z_i &= w_i\\
\dot v_i &= a_i\\
\dot \theta_i &= \phi_i\\
\dot w_i &= \delta_i,
\end{aligned}
\end{equation}
where $x_i$ and $y_i$ are the horizontal coordinates, $z_i$ is the vertical coordinate, $\theta_i$ is the horizontal heading angle, $v_i$ is the horizontal velocity, $w_i$ is the vertical speed, and $a_i$, $\phi_i$ and $\delta_i$ are control inputs. By feedback linearizion, the horizontal dynamics can be converted to a double integrator model. Define $v_{xi} \triangleq v_i\cos\theta_i$, and $v_{yi} \triangleq v_i\sin\theta_i$. Then, we have:
\begin{equation}
\begin{aligned}
\begin{bmatrix}
\dot v_{xi}\\ \dot v_{yi}
\end{bmatrix} 
= \begin{bmatrix}
\cos\theta_i &- v_i\sin\theta_i\\
\sin\theta_i & v_i\cos\theta_i 
\end{bmatrix}
\begin{bmatrix}
 a_i\\ \phi_i
\end{bmatrix}. 
\end{aligned}
\end{equation}
Let $[u_{xi},\,u_{yi}] = [\dot v_{xi},\,\dot v_{yi}]$ be our new control input. We then have the following relation between the new and original control inputs: 
\begin{equation}
\begin{bmatrix}
a_i\\ \phi_i
\end{bmatrix} = \begin{bmatrix}
\cos\theta_i & \sin\theta_i\\
\frac{-\sin\theta_i}{v_i} & \frac{\cos\theta_i}{v_i}
\end{bmatrix}\begin{bmatrix}
u_{xi}\\u_{yi}.
\end{bmatrix} 
\end{equation}
The transformation is not defined for $v_i= 0$, which will not be the case for the fixed-wing sUAS. Now we have a double integrator dynamics for the fixed-wing sUAS:
\begin{equation}\label{double_int_system}
\begin{aligned}
\ddot x_i &= u_{xi}\\
\ddot y_i &= u_{yi}\\
\ddot z_i &= \delta_i.
\end{aligned}
\end{equation}
Finally we let $q_i \triangleq [x_i, y_i, z_i]^T$ be the state vector, and let the stack vector $q \triangleq \text{col}(q_1,\,q_2,\,\dots)$. The above model simplifies our process for designing control protocols. This also admits that our following approaches can be easily adopted to control protocol design for multi-copters whose dynamics  can be approximated by a double integrator.

\subsection{Problem Formulation}
Here we introduce the problem formulation for traffic regulation for a single link. A single link is defined as a tuple $L\triangleq(\Omega,\hat v, \hat d, d_{\min}, d_{b,\min}, \hat d_b)$, where $\Omega$ is the physical space the link takes, $\hat v$ is the desired velocity for all sUAS in $\Omega$, $\hat d$ is the desired separation between the sUAS, and $d_{\min}$ is the minimum separation allowed between sUAS. $\hat d_b$ is the desired distance to the boundary of the link, and $d_{b,\min}$ is the minimum separation to the boundary of the link. Note that $\hat v$ is a velocity vector. We assume that $\Omega$ is a convex polyhedron, i,e., $\Omega = \{x|Ax\leq b, \, \partial A x \leq \partial b \}$, where $A\in \mathbb{R}^{m \times 3}$ and $b\in \mathbb{R}^m$ representing the wall/boundary; $\partial A\in \mathbb{R}^{{m'}\times 3}$ and $\partial b\in \mathbb{R}^{m'}$ denote the entrance/exit. When the link is described by a rectangular tube, then $m = 4$, and $m'=2$. Let $A_n$ and $b_n$ be the $n_{th}$ row of $A$ and $b$, respectively. Denote the distance from the $i_{th}$ sUAS to the plane $A_nx =  b_n$ as $d_{in}$. Then, $d_{in}$ can be given as:
\begin{equation}
d_{in} = \frac{A_nq_i-b_n}{A_nA_n^T}.
\end{equation}
We assume that the set
\begin{equation}
\mathcal A \triangleq\bigcap_{n = 1}^m\{x|\frac{A_nx-b_n}{A_nA_n^T}\leq \hat d_b\}
\end{equation}
 is not empty.
It is clear that the reference velocity should be parallel to each plane, i.e., $\forall n,\,\hat{v}^TA_n=0,$.
We define the set:
\begin{equation}
\mathcal I_{\Omega}(t) \triangleq \{i\in\mathcal I|\exists \tau\in[t_0, t], q_i(\tau)\in \Omega\}
\end{equation}
as the set of all the sUAS which entered the link up to time $t$. The objective of sUAS traffic control can be divided into two cases. First, if there is no sUAS entering $\Omega$ after $t_0$, i.e., $\forall t>t_0,\;\mathcal I_{\Omega}(t) = \mathcal I_{\Omega}(t_0)$. We have the following objectives.
\begin{equation*}
\begin{aligned}
O_1:&\forall i\in \mathcal I_{\Omega}(t_0),\;\dot q_i(t)\to \hat v\;\text{as }t\to\infty,\\
O_2:&\forall i,\,j\in \mathcal I_{\Omega}(t_0),\;||q_i(t)-q_j(t)||\geq \hat d\;\text{as }t\to\infty,\\
O_3:&\forall i\in \mathcal I_{\Omega}(t_0),\;\forall n= 1\dots m,\; d_{in}\geq \hat d_b \;\text{as }t\to\infty,
\end{aligned}
\end{equation*}
subject to the following constrains:
\begin{equation*}
\begin{aligned}
C_1:&\forall i\in \mathcal I_{\Omega}(t_0),\;||\dot q_i(t)-\hat v||\in[\underline v,\overline v]\;\forall t\geq t_0,\\
C_2:&\forall i,\,j\in \mathcal I_{\Omega}(t_0),\;||q_i-q_j||\geq d_{\min}\;\forall t\geq t_0,\\
C_3:&\forall i\in \mathcal I_{\Omega}(t_0),\;\forall n= 1\dots m,\;d_{in}\geq d_{b,\min} \;\forall t\geq t_0.
\end{aligned}
\end{equation*}
$O_1$ requires the velocities of all the sUAS in the link converge to the desired velocity. $O_2$ requires that the separations of all the sUAS are greater than the desired separation given sufficiently long time. $O_3$ requires that the positions of all the sUAS in the link converge to the desired separation from the boundary.  $C_1$ requires the boundedness of the velocities of all the sUAS: the upper bound is given from traffic authority, and the lower bound is required for the flyable trajectory for fixed-wing sUAS. $C_2$ requires that the separations between sUAS must be greater than or equal to the minimum separation to ensure collision free. $C_3$ requires that all the sUAS in the link must stay away from the boundary greater than or equal to minimum distance. Second, if there are sUAS entering $\Omega$ at time $t_1$, let $\mathcal I_\Omega(t_1^-)$ denote the set of the sUAS already in the link up to time $t_1$, i.e.,
\begin{equation}
    \mathcal I_{\Omega}(t^-) \triangleq \{i\in\mathcal I|\exists \tau\in[t_0, t), q_i(\tau)\in \Omega\}.
\end{equation}
Then we have $\mathcal I_\Omega(t_1^-)\subsetneq \mathcal I_\Omega(t_1)$. In this case, only $C_1$, $C_2$, and $C_3$ need to be guaranteed under some conditions on the entry states which will be discussed. It should be remarked that sUAS $i$ might be in the link initially and leaves the link at $t^* > t_0$, but by our definition of $\mathcal I_{\Omega}(t)$, $i \in  \mathcal I_{\Omega}(t)$ for any $t > t_0$. It is equivalent to assume the link is infinitely long such that whenever an sUAS enters it, it stays in it. This assumption facilitates the problem formulation and analysis without loss of generality, and it can be relaxed based on the approach in this work.

\subsection{Control Protocol Design and Analysis}
According to the control objectives and constraints, an APF-based control protocol will be introduced in this subsection, which ensures $O_1$-$O_3$ and $C_1$-$C_3$ are satisfied assuming no sUAS enters the link $\Omega$, i.e., 
\begin{equation}
    \mathcal I_\Omega(t) = \mathcal I_\Omega(t_0), \ \ \forall t > t_0.
\end{equation}
The convergence rate of the control protocol will be discussed, based on which the entry condition will be established in the next subsection.

For smooth artificial potential field design, the $\sigma$-norm function $||\cdot||_\sigma: \mathbb{R}^n\to \mathbb{R}_+$ is commonly considered~\cite{olfati2006flocking}:
\begin{equation}
||z||_\sigma = \frac{1}{\varepsilon}(\sqrt{1+\varepsilon||z||^2}-1),
\end{equation}
where $\varepsilon>0$ is a parameter, and $||\cdot||$ is the Euclidean norm. $\sigma$-norm is an approximation for the Euclidean norm but equipped with the differentiability at $z = 0$. The gradient of $\sigma$-norm is given as:
\begin{equation}
\nabla ||z||_\sigma = \frac{z}{\sqrt{1+\varepsilon||z||^2}}.
\end{equation}
For a nonzero vector $d$, denote $d_\sigma = ||d||_\sigma$, and it can be shown that $||d||>d_\sigma$.
We let the monotone function $\psi:\mathbb R^+\to \mathbb R^+$ be the repulsive potential function for collision avoidance, and $\phi:\mathbb R^+\to \mathbb R$ be the gradient of $\psi$. $\psi$ and $\phi$ satisfy:
\begin{equation}\label{psi_prop}
\begin{aligned}
\psi(d) = 0 &\iff d\geq \hat d\\
\phi(d) = 0 &\iff d\geq \hat d\\
\psi(d) >0 &\iff 0\leq d< \hat d.
\end{aligned}
\end{equation}
Let $V_p$ be the accumulated collision potential energy:
\begin{equation}
V_p = \frac{1}{2}\sum_{i}\sum_{j\neq i}\psi(||q_{ij}||_\sigma),
\end{equation}
where
\begin{equation}
    q_{ij} = q_i - q_j,
\end{equation}
then $V_p$ achieves the global minimum of $0$ at $\forall i,j, \in \mathcal{I}_{\Omega}(t_0)$ with $i \neq j$, $||q_{ij}||_\sigma \geq \hat d$, thus, $||q_{ij}||>\hat d$. Similarly, we let $\psi_b$ be the potential function for boundary separation and $\phi_b$ be the gradient of $\psi_b$ that (\ref{psi_prop}) is satisfied with $\hat d$ replaced by $\hat d_b$. Let $V_b$ be the accumulated boundary potential energy:
\begin{equation}
V_b = \sum_i\sum_n \psi_b(d_{in}).
\end{equation}
$V_b$ achieves the global minimum of $0$ when $\forall i \in \mathcal{I}_{\Omega}(t_0)$, $\forall k \in  \left\{1,...,m \right\}$, $d_{in} \geq \hat d_b$
We define the accumulated kinetic energy as:
\begin{equation}
V_k = \frac{1}{2}\sum_{i}(\dot q_i - \hat v)^T(\dot q_i - \hat v).
\end{equation}
It is clear that $V_k$ achieves the global minimum of $0$ if $\forall i \in \mathcal{I}_{\Omega}(t_0),$, $\dot q_i = \hat v$.
Then, the feedback control protocol is designed as:
\begin{equation}\label{control_law}
\begin{aligned}
u_i& = -\sum_{j\neq i}\phi(||q_{ij}||_\sigma) \hat{q}_{ij}-\sum_{n=1}^m\phi_b(d_{in})A_n^T- K_i(v_i-\hat{v}),\\
\end{aligned}
\end{equation}
where:
\begin{equation}
\begin{aligned}
\hat{q}_{ij} &= \frac{q_{ij}}{\sqrt{1+\varepsilon||q_{ij}||^2}-1},
\end{aligned}
\end{equation}
where $K_i$ is the damping ratio, which is a tuning parameter whose design needs to address the physical capability of the sUAS. We can define the positive semi-definite Hamiltonian function:
\begin{equation}
\begin{aligned}
H &= V_p+V_k+V_b.
\end{aligned}
\end{equation}
The following theorem shows $O_1$, $O_2$ and $O_3$ are achieved given that constraints $C_1$, $C_2$ and $C_3$ are satisfied all the time under the case where no sUAS enters the link after $t_0$.
\begin{theorem}\label{main}
Consider a fixed group of sUAS (\ref{dynamics}) in the link $\Omega$, with the initial configuration in $\Omega_c =\{(q(t_0),\dot q(t_0))|H(t_0) = c\}$ applied with control protocol (\ref{control_law}) for $t>t_0$. The following hold:\\
(i) $\dot H(t)\leq 0,\,\forall \,t>t_0$.\\
(ii) Almost every solution of the multi-sUAS dynamics governed by (\ref{dynamics}) converges to an equilibrium where $\forall i,j \in \mathcal{I}_{\Omega}(t_0)$ with $i \neq j,$, $\dot q_i = \hat v,\,d_{in}\geq \hat d_b,\,||q_{ij}||\geq \hat d$ as $t\to\infty$. \\
(iii) If $c \leq c_1^* \triangleq \psi(||d_{\min}||_\sigma)$, no pair of sUAS violate the minimum separation $\forall t>t_0$. \\
(iv) If $c\leq c_2^* \triangleq \frac{1}{2}\tilde v^2$, where $\tilde v = \min\{\overline v-||\hat v||,\,||\hat v||-\underline v\}$, then $\forall i \in \mathcal{I}_{\Omega}(t_0),\,t>t_0,\,||\dot q_i||\in[\underline v, \overline v].$\\
(v) If $c\leq c_3^* \triangleq \phi_b(d_{b,\min})$, then $\forall i \in \mathcal{I}_{\Omega}(t_0), t>t_0,\,d_{in}>d_{b,\min}$.
\end{theorem}
\begin{proof}
(i) Denote $\delta v_i \triangleq \dot q_i - \hat v$. We then have:
\begin{equation}\label{H_dot}
\begin{aligned}
\dot H(q,\dot q) &= \sum_{i}\delta v_i^Tu_i+\frac{d}{dt}\Big[\frac{1}{2}\sum_{j\neq i}\psi(||q_{ij}||_\sigma)+\sum_{n=1}^{m} \psi_b(d_{in})\Big]\\
\end{aligned}
\end{equation}
By definition, we have:
\begin{equation}
\begin{aligned}
\frac{d}{dt}\psi(||q_{ij}||_\sigma) &= \phi(|q_{ij}|)\hat{q}_{ij}{}^T(\dot q_i-\dot q_j)\\
\frac{d}{dt}\psi_b(d_{in}) &= \phi_b(d_{in})A_n{}^T\dot q_i
\end{aligned}
\end{equation}
Substituting control protocol (\ref{control_law}) into (\ref{H_dot}) yields:
\begin{equation}
\dot H = \sum_{i}-K_i\delta v_i^T\delta v_i.
\end{equation}
(ii) Statement (i) implies that $\dot H \leq 0$ given that $H$ is positive semidefinite. From LaSalle's invariant principle, all the solutions converge to the largest invariant set contained in $\mathcal I  = \{(q,\dot q)|\dot H = 0\}$. Starting from almost every initial configuration, the positions of all sUAS eventually satisfy the desired separation and their velocities match the reference velocity. \\
(iii) By contradiction, suppose there exists $t_1>t_0$, sUAS $i^*$ and $j^*$ collide, i.e., $||q_{i^*j^*}(t_1)||_\sigma\leq ||d_{\min}(t_1)||_\sigma$, and then
\begin{equation}
\begin{aligned}
H(t_1) &\geq \psi(||q_{i^*j^*}||_\sigma) + \frac{1}{2}\sum_{i\neq i^*,j^*}\sum_{j\neq i,i^*,j^*} \psi(||q_{i^*j^*}||_\sigma)\\
&\geq c_1^*>c.
\end{aligned}
\end{equation}
By Statement (ii), $\forall t>t_0,\;\dot H(t)\leq 0$, therefore $H(t_1)\leq H(t_0)=c<c_1^*$. Absurd.\\
The proof for Statement (iv) and (v) follows the same argument as the proof for Statement (iii), therefore omitted.
\end{proof}
By letting the level set for initial configuration as $\{(q(t_0),\dot q(t_0)|H(t_0)\leq c^*\}$, where $c^* = \min\{c^*_1,c^*_2,c^*_3\}$. One can achieve $O_1,\,O_2,\,O_3$ for almost every initial configuration while $C_1, \,C_2, \,C_3$ are satisfied.

It can be seen that under the adoption of LaSalle's invariant principle, the asymptotic convergence is given but without a convergence rate. The convergence rate of the Hamiltonian can be crucial for designing the entry condition for a link or estimating the behaviors of sUAS in a link for any given time instance. Now we study the convergence property for a fixed group of multi-sUAS system (\ref{double_int_system}) with control protocol (\ref{control_law}). We start with the augmented system dynamics:
\begin{equation}\label{augmented_hamiltonian}
	\ddot q + K\dot q + \nabla \Psi(q) = 0,
\end{equation}
where $K$ is the collective damping ratio/time constant, and $\Psi$ is the collective potential function. It can be seen from (\ref{control_law}) that $K$ is a diagonal matrix with positive diagonal entries. Without loss of generality, we assume the diagonal elements of $K$ are the same, such that $K$ can be reduced to scalar. We will refer $K$ as a scalar in the proceeding without further notification. Such system is often referred as the gradient Hamiltonian system. The asymptotic convergence of (\ref{augmented_hamiltonian}) of different types has been extensively discussed; however, the convergence rate is rarely given~\cite{tanaka1991homoclinic,kozlov2003weak,broyden1969new,helmke2012optimization,alvarez2000minimizing}. It is noted that system (\ref{augmented_hamiltonian}) can also be viewed as a second order differential equation method for solving the following optimization problem:
 \begin{mini}|l|
 	{q}{\Psi(q)}{}{}.
 \end{mini}
The convergence rate is still rarely discussed under the continuous time optimization framework or flocking control framework. 
Motivated by the lack of convergence rate analysis, we develop following the convergence result for the general Hamiltonian system. We start with following assumptions:
\begin{assumption}\label{L-smooth}
$\nabla \Psi$ is differentiable and Lipschitz with constant $L$.
\end{assumption}
\begin{assumption}\label{KL-ineq}
$K>\sqrt{\frac{1}{4L}}$.
\end{assumption}
Our result for convergence rate is given in the following theorem:
\begin{theorem}\label{convergence-rate}
For system (\ref{augmented_hamiltonian}), if Assumptions \ref{L-smooth}, \ref{KL-ineq} hold and $\Psi(x(0)), \nabla \Psi(x(0)), \dot x(0)$ are finite, then:
\begin{equation}
\begin{aligned}
\inf_{\tau \in (0,t)}& p\nabla \Psi(q(\tau))^T\nabla\Psi(q(\tau))+r\dot q(\tau)^Tq(\tau)\leq\\ &\frac{1}{t}\Big(\frac{1}{2}||\alpha\nabla\Psi(q(0))+\dot x(0)||^2_2+(\alpha+K)\Psi(q(0))\Big)
\end{aligned}
\label{eq:th2eq}
\end{equation}
where $p = \frac{2KL-K^2}{4L^2}$, $r = \frac{K}{4}$, and $\alpha = \frac{K}{2L}$.
\end{theorem}
\begin{proof}
Given the dynamic system (\ref{augmented_hamiltonian}), let Lyapunov-like functions be:
\begin{equation}
    \begin{aligned}
        V_1 &= \frac{1}{2}||\alpha\nabla\Psi(q)+\dot q||_2^2\\
        V_2 &= (\alpha k + 1)\Psi(q)\\
        V_3 &= \int_{0}^{t}p\nabla\Psi(q)^T\nabla\Psi(q) + r\dot q^T\dot q dt.\\
        V &= V_1 + V_2 +V_3.
    \end{aligned}
\end{equation}
Thus $V_1$, $V_2$, and $V_3$ are positive semi-definite. The time derivative of $V_1$ is:
\begin{equation}
    \begin{aligned}
        \dot V_1 =& (\alpha\nabla\Psi(q)+\dot q)^T(\alpha \nabla^2\Psi(q)\dot q - K\dot q-\nabla \Psi(q))\\
        =& -\dot q^T(KI-\alpha\nabla^2\Psi(q))\dot q-\alpha\nabla \Psi^T\nabla\Psi(q)\\
        &+\nabla\Psi(q)^T(\alpha^2\nabla^2\Psi(q)-\alpha KI -I)\dot q.
    \end{aligned}
\end{equation}
The time derivative of $V_2$ and $V_3$ are respectively:
\begin{equation}
    \begin{aligned}
        \dot V_2 =& (\alpha k+1)\nabla \Psi(q)^T\dot q\\
        \dot V_3 =& p\nabla \Psi(q)^T\nabla\Psi(q) + r\dot q^T\dot q.
    \end{aligned}
\end{equation}
Adding the derivatives gives:
\begin{equation}
    \begin{aligned}
        \dot V =& -\dot q^T((K-r)I-\alpha \nabla^2 \Psi(q))\dot q \\
        &- (\alpha - p)\nabla \Psi(q)^T\nabla\Psi(q)+\alpha^2\nabla\Psi(q)^T\nabla^2\Psi(q)\nabla\Psi(q)
    \end{aligned}
\end{equation}
Given $\Psi(q)$ is $L$-smooth, we have $\nabla^2\Psi\leq LI$. Under Assumption \ref{KL-ineq}, by letting $r = \frac{K}{4}$ and $\alpha = \frac{K}{2L}$, we have:
\begin{equation}
    \begin{aligned}
        \dot V \leq& -\frac{K}{4}\dot q^T\dot q- (\alpha - p)\nabla \Psi(q)^T\nabla\Psi(q)\\
        &+\alpha^2\nabla\Psi(q)^T\nabla^2\Psi(q)\nabla\Psi(q)\\
        =&  -\frac{K}{4}(\dot q^T\dot q+\frac{K}{L^2}\nabla\Psi(q)^T\nabla^2\Psi(q)\dot q\\
        &-(\frac{2}{L}-\frac{p}{k})\nabla\Psi(q)^T\nabla\Psi(q)).
    \end{aligned}
\end{equation}
By letting $p = \frac{2KL-K^2}{4L^2}$, we have 
\begin{equation}
    (\frac{2}{L}-\frac{p}{k})I \geq (\frac{K}{2L})^2\nabla^2\Psi(q)^T\nabla^2\Psi(q).
\end{equation}
Thus, 
\begin{equation}
    \dot V\leq -\frac{K}{4}||\dot q+ \frac{K}{2L^2}\nabla^2\Psi(q)\nabla(q)||_2^2\leq 0.
\end{equation}
By the monotonicity of $V$ and the positive semi-definiteness of $V_1$ and $V_2$, We have:
\begin{equation}
    \begin{aligned}
        \int_{0}^{t}p\nabla\Psi(q)^T\nabla\Psi(q) + r\dot q^T\dot q dt \leq V_1(0)+V_2(0),
    \end{aligned}
\end{equation}
which implies (\ref{eq:th2eq}).
\end{proof}

\begin{remark}
The significance of Theorem \ref{convergence-rate} does not limit to UTM applications. It is the first attempt for establishing the convergence rate result for a general class of Hamiltonian systems, as well as for the continuous time algorithm for solving non-convex optimization problem. It is not surprising that $O(\frac{1}{t})$ convergence is achieved, since such speed is well established for the first order methods for solving non-convex optimization problem. 
This result can be applied to the analysis for the behavior of general classes of controlled system, such as flocking control, nonlinear control for dissipative systems, to name a few. It can offer more interpretations for the system behavior under different contexts while specific understanding or structure of the system is present. 
\end{remark}

Note that Theorem \ref{convergence-rate} is not directly useful in our problem because it only provides an upper bound for the velocities and gradient forces, not for the Hamiltonian function. The following proposition provides an insight for convergence rate of Hamiltonian function under a proper assumption:

\begin{proposition}
 Consider system (\ref{augmented_hamiltonian}) with Assumptions \ref{L-smooth}, \ref{KL-ineq}. Suppose there are positive numbers $\underline \alpha$, $\overline \alpha$ for the trajectory $q(t)$ such that
 \begin{equation}
        \underline \alpha \Psi(q) \leq \nabla \Psi(q){}^T \nabla \Psi(q) \leq \overline \alpha  \Psi(q).
        \label{eq:DWassum}
 \end{equation}
 Then there is a positive number $\lambda$ such that
 \begin{equation}
     H(t) \leq \frac{\lambda}{t-t_0} H(t_0).
 \end{equation}
 \label{PR1}
\end{proposition}

\begin{proof}
 The proof is purely algebraic since both the left-hand-side and right-hand-side of (\ref{eq:th2eq}) involve the quadratic terms of $\nabla \Phi$ and $\dot q$. ``$inf$" is removed since $H$ is monotonically decreasing.
\end{proof}

\begin{remark}
It is clear that when the sUAS system is at a desired configuration, condition (\ref{eq:DWassum}) holds. For $q$ near the desired configuration, condition (\ref{eq:DWassum}) can also be achieved by choosing proper artificial potential functions (e.g., locally quadratic functions near the desired configuration). The existence of $\underline \alpha$ and $\overline \alpha$ allows us to establish a direct convergence bound on the total Hamiltonian of the system. For $\lambda$, it can be either algebraically derived or estimated via test or numerical experiments if a less conservative result is desired. The inequality (\ref{eq:DWassum}) may not be necessary for the convergence rate to hold, as one can see in the simulation results. Such results allow us to offer theoretical completeness for the proceeding results on sUAS entry condition design. 
\end{remark}

\subsection{Entry Condition Design}
In the last subsection, we have designed a control protocol based on APF and we have shown that the state of a fixed group of sUAS converges to an invariant set in which all sUAS maintain the desired separation and track the desired speed for almost every initial configuration. The convergence rate of the Hamiltonian is proved for a fixed group of sUAS. The reason for deriving this rather stronger convergence property is to quantify the speed of regulating the behaviors of the sUAS in a link. For almost every initial condition, we are able to bound the total energy/Hamiltonian in a link at a given time instance. It allows us to derive the time when the system is ready to accept more sUAS into the link and to estimate the total amount of energy the incoming sUAS can bring into the system. By bounding the energy carried by the incoming sUAS for each fixed time interval, we are able to bound the total number of incoming sUAS under assumptions. In this way, our framework can not only account for the micro level regularization of the sUAS behaviors, but also offer theoretical characterization to the macro level traffic property, in particular, the flow rate of each traffic link. Such characterization allows further evaluation of the efficiency of the whole traffic network. 

Let
\begin{equation}
\{t|\mathcal I_\Omega(t)\subsetneq \mathcal I_\Omega(t^-)\}
\end{equation}
be the set of time instances when sUAS enter $\Omega$. This set is called the set of time instances of \textit{regular entry}. Denote the $k$-th entry instance as $t_k$. To limit the entry frequency for sUAS, we consider the following assumption:
\begin{assumption}\label{ASentry_time}
$\forall k\in\mathbb{Z}_{\geq 0},\,t_{k+1}-t_{k}\geq T.$
\end{assumption}
According to Proposition \ref{PR1}, we have
\begin{proposition}
Consider the sUAS in a link with Assumption \ref{entry_time}. Assume the conditions in Proposition \ref{PR1} hold. Suppose $H(t_0) < h_0$ and for each entry instance $t_k$, $H(t_k) - H(t_k^-) \leq h_\epsilon < h_0$. If $T$ given in the Assumption \label{entry_time} satisfies
\begin{equation}
    T \geq \frac{\lambda h_0}{h_0 - h_\epsilon},
\end{equation}
then $H(t) \leq h_0$, $\forall t \geq t_0$. 
\label{PR2}
\end{proposition}
Such a result is desirable in the sense that if the entry rate is bounded, the Hamiltonian function value is bounded such that collision and boundary violations are excluded. Note that our definition of entry event allows multiple sUAS enter the link simultaneously and Proposition \ref{PR2} still applies. However, the result will be conservative considering the cases where there may be multiple sUAS enter the link intermittently in a short time period with relatively regulated configuration. To address such cases, we consider the following definition for the intermittent entry event of multiple sUAS in a short time period.

$t_k$ is said to be a time instance for a \textit{multiple entry} event if there are multiple sUAS enter the link at $t_k$ and before $t_k+t_\epsilon$, where $t_\epsilon$ is a design parameter that should be small. Denote $\mathcal I_{\partial\Omega}(t_k)$ as the index set for the sUAS enter the link during $[t_k, t_k + t_\epsilon)$. Let $t_{k+1}-(t_k+t_\epsilon)\geq T$ and $T>>t\epsilon$. Since $t_\epsilon$ is small, it is of interest to establish entry conditions such that the sUAS traffic behavior satisfies the constraints during $[t_k+\epsilon, t_{k+1})$. It is done by the following proposition:

\begin{proposition}
\label{entry_lemma}
 Consider sUAS with dynamics (\ref{dynamics}) in the link $\Omega$ with conditions in Proposition \ref{PR1} valid. Let $t_k$, $k \in \mathbb{Z}_{\geq 0}$ be the time instances of entry events (either regular entry or multiple entry) with Assumption \ref{ASentry_time} valid. For all $k \in \mathbb{Z}_{\geq 0}$, assume $\forall i\in \mathcal I_{\partial\Omega}(t_{k+1}),\;j\in\mathcal I_\Omega(t_{k+1}^-),\;\psi(||q_{ij}(t_{k+1}+t_\epsilon)||_\sigma) = 0$. Let
 \begin{equation}\label{eq:kappa_def}
 \begin{aligned}
 \kappa &= \max_{i\in \mathcal I_{\partial\Omega}(t_{k+1}+t_\epsilon)} \frac{1}{2}||\delta v_i(t_{k+1}+t_\epsilon)||_2^2+\sum_{n=1}^m \psi_b(d_{in}) \\
 \gamma &= \max_{i \neq j\in \mathcal I_{\partial\Omega}(t_{k+1}+t_\epsilon)} \frac{1}{2}\psi(||q_{ij}(t_{k+1}+t_\epsilon))||_\sigma).
 \end{aligned}
 \end{equation}
 If
 \begin{equation}\label{eq:numberieq}
     M\kappa +M(M-1)\gamma\leq c^*(1-\frac{\lambda}{T}),
 \end{equation}
where $M$ is a upper bound for the number of sUAS for all entry events, and then $\forall t \in \cup_{k} [t_{k}+t_\epsilon,t_{k+1})$, $C_1$, $C_2$, $C_3$ are satisfied.
\end{proposition}


\begin{proof}
 By the given conditions, we have:
\begin{equation}
\begin{aligned}
H(t_{k+1}+t_\epsilon) =& \sum_{i\in \mathcal I_{\Omega}(t_{k+1}^-)}\frac{1}{2}\delta v_i^T\delta v_i +\sum_{n=1}^{m}\psi_b(d_{in})\\
&+\frac{1}{2}\sum_{j\in \mathcal I_{\Omega}(t_{k+1}^-),j\neq i}\psi(||q_{ij}||_\sigma)\\
&+\sum_{i\in \mathcal I_{\partial\Omega}(t_{k+1}+t_\epsilon)}\frac{1}{2}\delta v_i^T\delta v_i + \sum_{n=1}^{m}\psi_b(d_{in})\\
&+\sum_{{i\in \mathcal I_{\Omega}^\partial(t_{k+1}+t_\epsilon)},j\neq i}\psi(||q_{ij}||_\sigma)
\end{aligned}
\end{equation}
for any $k \in \mathbb{Z}_{\geq 0}$. Given that $H(t_k)\leq c^*$, and by Proposition \ref{PR1}, we have:
\begin{equation}
\begin{aligned}
H(t_{k+1}^-) \leq& \frac{c^*\lambda}{T}.
\end{aligned}
\end{equation}
Thus,
\begin{equation}
\begin{aligned}
H(t_{k+1}+t_\epsilon)\leq& \sum_{i\in \mathcal I_{\partial\Omega}(t_{k+1}+t_\epsilon)}\frac{1}{2}\delta v_i^T\delta v_i +  \sum_{n=1}^{m}\psi_b(d_{in}) \\
+&\frac{1}{2}\sum_{{j\in \mathcal I_{\partial\Omega}(t_{k+1}+t_\epsilon)},j\neq i}\psi(||q_{ij}||_\sigma)\\
&+\frac{c^*\lambda}{T}\\
\leq & M\kappa+M(M-1)\gamma + \frac{c^*\lambda}{T}.
\end{aligned}
\end{equation}
Given $M\kappa +M(M-1)\gamma\leq c^*(1-\frac{\lambda}{T})$, we conclude that $H(t_{k+1}+t_\epsilon)\leq c^*$. By Theorem \ref{main}, $C_1$, $C_2$, and $C_3$ are satisfied $\forall t \in \cup_{k} [t_{k}+t_\epsilon,t_{k+1})$.
\end{proof}
\begin{remark}
The above proposition assumes that the distance between any sUAS entering the link at $t_{k+1}$ and any sUAS already in the link before $t_{k+1}$ is greater than or equal to the desired separation. Such assumption can be realized by the design of $T$ for any given $\hat v$. It is also assumed that the boundedness of the Hamiltonian is satisfied all the time before the entry of the group of sUAS. This assumption shall not be violated in practice to ensure safety.
\end{remark}
Proposition \ref{entry_lemma} serves a practical tool for entry condition regularization and link transition design. It is intuitive that two very ``different" links should not be connected. The following discussion quantifies the maximum difference between two links in order to be connected. Let the upstream link and downstream link be $L_1 = (\Omega_1,\hat v_1, \hat d_1, d_{1\min},\hat d_{1b}, d_{1b,\min})$ and $L_2 = (\Omega_2,\hat v_2, \hat d_2, d_{2\min},\hat d_{2b}, d_{2b,\min})$, respectively. Let $\psi^1,\,\psi_b^1$ and $\psi^2,\psi_b^2$ be the potential energy functions used for collision avoidance and boundary separation for $L_1$ and $L_2$, respectively. Assume that the $L_1$ is long enough that the exiting sUAS are at desired configuration. We also define the event of link transition. sUAS $i$ is said to transit from $L_1$ to $L_2$ if:
\begin{equation}
q_i\in\mathcal A_1\cap \mathcal A_2.
\end{equation}
This admits that $\mathcal A_1\cap \mathcal A_2\neq \emptyset$, and $\psi_b^1(d_{in}) = \psi_b^2(d_{in}) = 0$ for $n = 1,\dots,m$. After a transition, sUAS will start using the control protocol defined on $L_2$. Then $\kappa$ and $\gamma$ in (\ref{eq:kappa_def}) can be quantified based on the definitions of $L_1$ and $L_2$, that is:
\begin{equation}
    \begin{aligned}
        \kappa &= \frac{1}{2}(\hat{v}_1-\hat{v}_2)^T(\hat{v}_1-\hat{v}_2)\\
        \gamma &= \frac{1}{2}\psi^2(\hat d_1)
    \end{aligned}
\end{equation}
It can be observed from (\ref{eq:numberieq}) that if $\kappa\geq (1-\lambda/T)c^*$, then no more than one sUAS can enter $L_2$ from $L_1$ to ensure constrains satisfaction. This phenomenon explains that a high speed link should not be connected to a low speed link immediately, but another transitional link will be necessary. The idea of the transitional link resembles the on/off ramp of the highway in ground traffic. The maximum number of sUAS that can enter $L_2$ from $L_1$ in every $[t_k,t_k+t_\epsilon)$ is then the maximum integer solution to inequality (\ref{eq:numberieq}). This identity connects our analysis on the entry condition to the marco flow control between two connected links.

\begin{figure*}[t!]
	\centering
	\label{T_sUAS}
    \includegraphics[width = \linewidth]{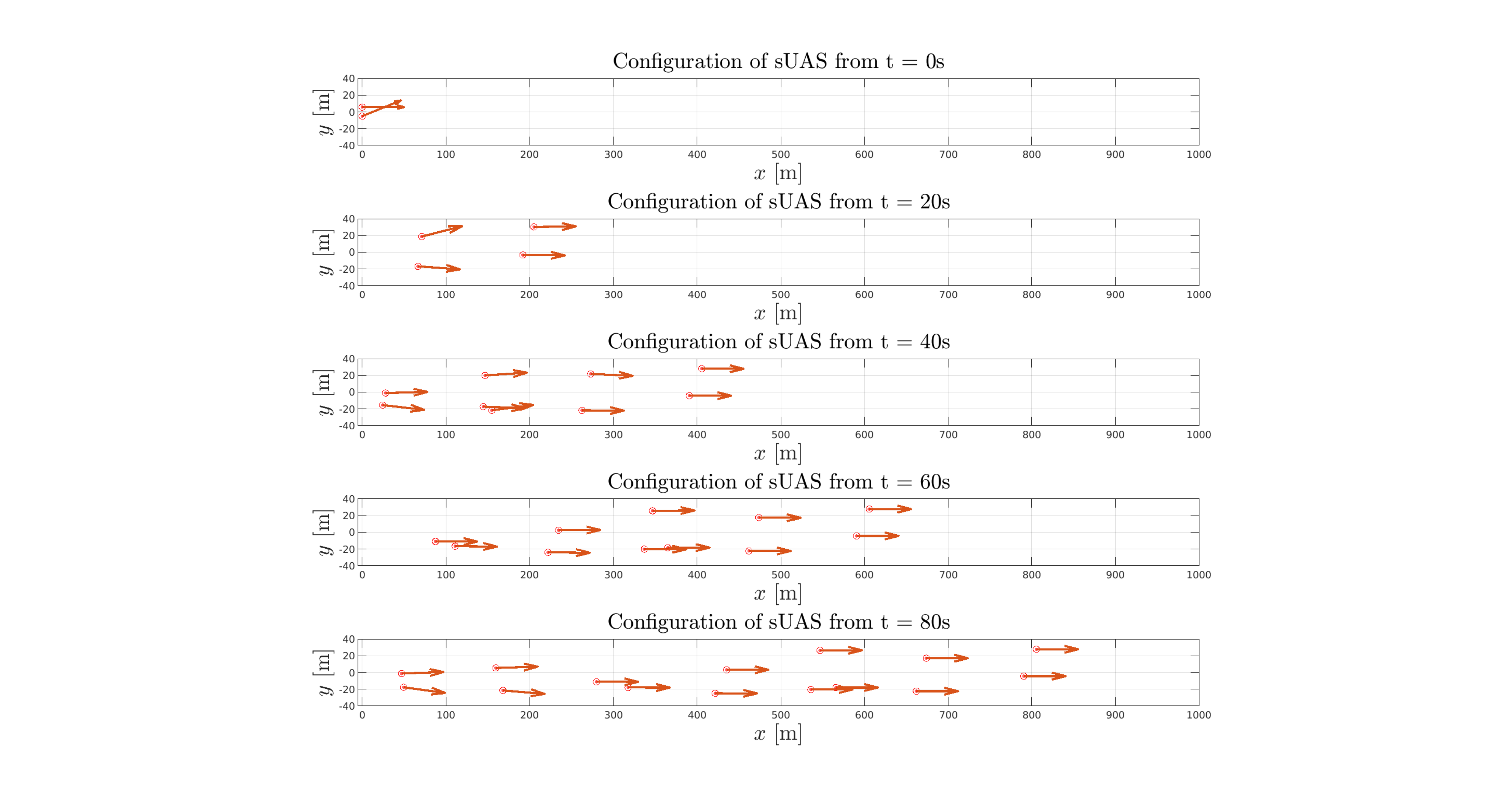}
	\caption{Configurations of sUAS in every 20s after the entry of a new group of sUAS}
	\label{T_sUAS}
\end{figure*}

\subsection{Hardware/Software Requirement}
In this section, we briefly discuss the hardware and software requirements of the implementation of our framework.  Based on the control protocol we designed, we now propose hardware and software requirements on the sUAS. First, under our current framework, a receiver is required on sUAS to receive the broadcast control protocol in each traffic link. Certain encoding and decoding protocols need to be designed in this process to prevent any potential cyberattack. The control protocol (\ref{control_law}) first assumes that position and velocity are estimated. This can be achieved by the integration of inertia measurement sensors, such as gyroscope, accelerometer, along with a Global Positioning System, which are typical equipped on an sUAS. Sensor fusion and state estimation techniques can be applied to achieve accurate state information for each sUAS. The subject of state estimation for highly non-linear sUAS dynamics is non trivial, and we refer to~\cite{rigatos2012nonlinear,bibuli2009path} for more details. Second, for the inter-sUAS separation, the relative distance measurement is required. Unlike classic formation control or flocking control problems where a cooperative information exchange protocol is assumed, it may not be practical to assume that all sUAS traveling in a traffic network can send and receive state information directly from a communication network. Furthermore, the relative position is required for our control protocol to be implemented. The existence of other sUAS and obstacles are not differentiated, and obstacle sensing techniques can be used to achieve collision avoidance and desired separation. Detection sensors like radar, LIDAR, or sonar can be used to measure relative distances of approaching objects. Combinations of sensors and data fusion techniques can be used to get more accurate state estimates. A comprehensive survey on sUAS sensing technologies can be found in~\cite{yu2015sense}.

\section{numerical experiments}\label{simulation-section}
In this section, we demonstrate the proposed control protocol for UTM with an illustrative numerical simulation. In the simulation, the potential function we chose is:
\begin{equation}
\psi(x) = \begin{cases}
\log(\cosh(x-\hat d))\; &\text{if } x<\hat d\\ 
0\;&\text{otherwise}
\end{cases}
\end{equation}
the gradient of $\psi$ is given by:
\begin{equation}
\phi(x) = \begin{cases}
\tanh(x-\hat d)\; &\text{if } x<\hat d\\
0\;&\text{otherwise}
\end{cases}
\end{equation}
We let $\psi_b = \psi$ and $\phi_b = \phi$. Let $\Omega$ be the polyhedron whose walls are defined by:
\begin{equation}
\begin{aligned}
A = \begin{bmatrix}
0 & 1 & 0\\
0&-1&0\\
0&0&1\\
0&0&-1
\end{bmatrix}, 
b = \begin{bmatrix}
40\\-40\\40\\-40
\end{bmatrix}.
\end{aligned}
\end{equation}
Other parameters for numerical simulation are given as:
\begin{table}
\centering
\captionof{table}{Parameters for simulation} \label{parameters} 
\begin{tabular}[h!]{| c | c |}
\hline
 $\hat v$ & $[10,\,0,\,0]^T$ [m/s]\\
 $\overline v$ & 25 [m/s]\\
 $\underline v$ & 5 [m/s]\\
 $d_{\min}$ & 1.5 [m]\\
 $\hat d$ & 10 [m]\\
 $d_{b,\min}$ & 0 [m]\\
 $\hat d_{b}$ & 20 [m]\\
 $K_i$ & 0.1\\
 $\varepsilon$ & 0.9\\
 $T$ & 20 [s]\\
 \hline
\end{tabular}
\end{table}
Based on the parameters, we can achieve the feasible $c^* = 8.3069$. The configurations of sUAS for every $20$ seconds after a new group of sUAS are presented in Figure \ref{T_sUAS}.

In this experiment, the set $\mathcal A \triangleq \{(x,y,z)|-20\leq y \leq 20,\,-20 \leq z \leq 20,\,0\leq x\leq 1000\}$, and it be observed that the configuration converges. We also plot the minimum  pairwise separation between sUAS in the traffic link in Figure \ref{figsep}. It can can be observed that all the sUAS satisfy the minimum separation rule at all time.
\begin{figure}[h!]
\centering
\includegraphics[width = 0.8\linewidth]{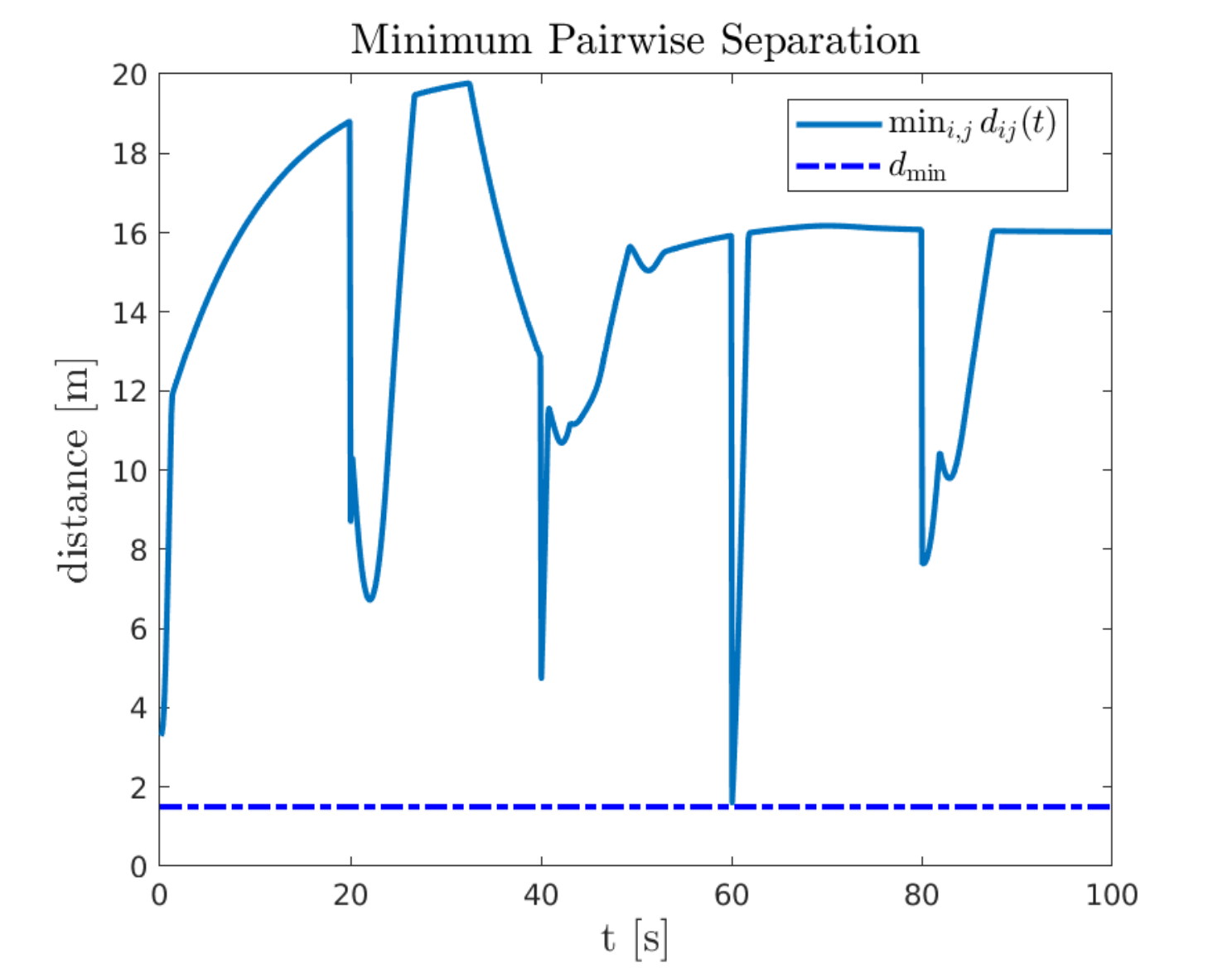}
\caption{Minimum Pairwise Separation}
\label{figsep}
\end{figure}
Similarly, it can be seen in Figure \ref{vel_bound} that the minimum and maximum velocity constraints are always satisfied. 
\begin{figure}
\centering
\includegraphics[width = 0.8\linewidth]{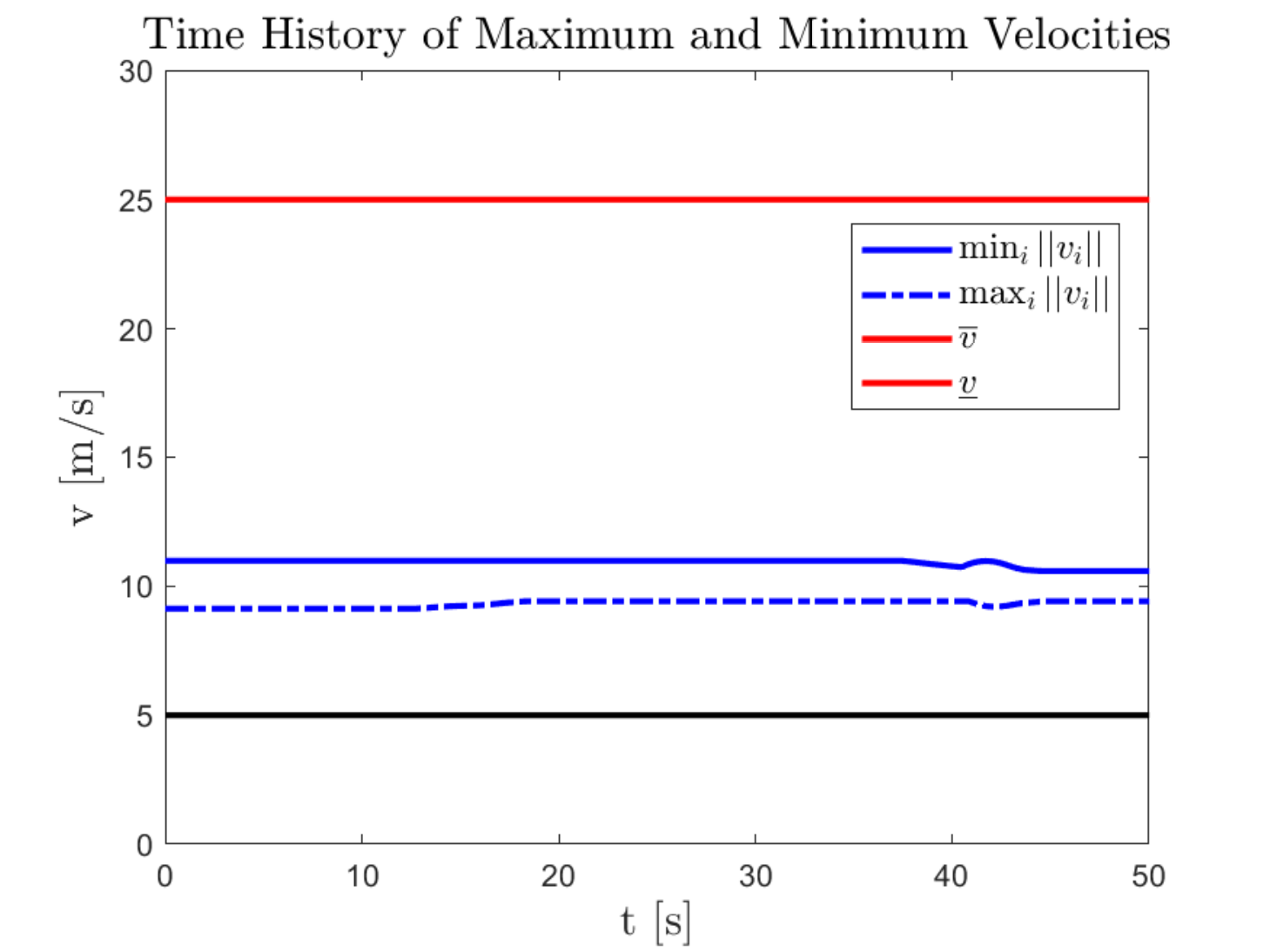}
\caption{Time history of Minimum Velocity and Maximum Velocity}
\label{vel_bound}
\end{figure}
This simulation shows that with our designed control protocol and entry condition, sUAS are able to enter a single link safely and regulate their speeds and separations. The fundamental benefit of our control protocol based sUAS traffic management is not limited to achieve a collective behavior, but increase the amount of traffic a single link can take under a short period of time. Under our simulation, at least 2 sUAS can enter the link at every 20 seconds, which allows 360 sUAS to travel across this link per hour. Such a great amount of traffic is formidable in the classical trajectory file/review process for traffic management if all the operators file their trajectory individually. The trajectory reviewing for this amount of traffic can also be overburden for UTM.

\section{Conclusion}~\label{conclusion}
In this article, we have proposed a new control protocol design and analysis method which can safely manage a large amount of sUAS, thereby improving scalability of UTM. By taking the benefits of the autonomous nature of the sUAS, the traffic management problem is reformulated as a distributed coordination control for multi-agent systems. We formally defined the sUAS behavior regularization problem in a single traffic link and proposed a control protocol to achieve control objectives without violating operating constraints. Further, we have analyzed the proposed control protocol and developed the condition for sUAS entering a traffic link. This entry condition could be successfully converted to traffic management criteria/rules. In the numerical experiments, the proposed control protocol has been shown to be effective, and the entry condition is validated. This work offers a fundamental framework and theoretical results for studying a micro-scope traffic regularization problem in more complex traffic network elements, such as merge links and split links.


\bibliographystyle{unsrt}
\bibliography{reference}

\begin{thebibliography}{10}

\bibitem{aweiss2018unmanned}
Arwa~S Aweiss, Brandon~D Owens, Joseph Rios, Jeffrey~R Homola, and Christoph~P
  Mohlenbrink.
\newblock Unmanned aircraft systems (uas) traffic management (utm) national
  campaign ii.
\newblock In {\em 2018 AIAA Information Systems-AIAA Infotech@ Aerospace}, page
  1727. 2018.

\bibitem{federal2016summary}
Federal Aviation~Administration (FAA).
\newblock Summary of small unmanned aircraft rule (part 107).
\newblock 2016.

\bibitem{johnson2017flight}
Marcus Johnson, Jaewoo Jung, Joseph Rios, Joey Mercer, Jeffrey Homola, Thomas
  Prevot, Daniel Mulfinger, and Parimal Kopardekar.
\newblock Flight test evaluation of an unmanned aircraft system traffic
  management (utm) concept for multiple beyond-visual-line-of-sight operations.
\newblock 2017.

\bibitem{NASAUTM}
National Aeronautics and Space Administration.
\newblock Utm: Air traffic management for low-altitude drones.
\newblock 2015.

\bibitem{NASAOStest}
National Aeronautics and Space Administration.
\newblock {NASA Conducts} ‘out of sight’ drone tests in nevada, 2015.

\bibitem{jiang2016unmanned}
Tao Jiang, Jared Geller, Daiheng Ni, and John Collura.
\newblock Unmanned aircraft system traffic management: Concept of operation and
  system architecture.
\newblock {\em International journal of transportation science and technology},
  5(3):123--135, 2016.

\bibitem{ren2017small}
Liling Ren, Mauricio Castillo-Effen, Han Yu, Yongeun Yoon, Takuma Nakamura,
  Eric~N Johnson, and Corey~A Ippolito.
\newblock Small unmanned aircraft system (suas) trajectory modeling in support
  of uas traffic management (utm).
\newblock In {\em 17th AIAA Aviation Technology, Integration, and Operations
  Conference}, page 4268, 2017.

\bibitem{ippolito2019autonomy}
Corey~A Ippolito, Kalmanje Krishnakumar, Vahram Stepanyan, A~Chakrabarty, and
  J~Baculi.
\newblock An autonomy architecture for high-density operations of small uas in
  low-altitude urban environments.
\newblock In {\em 2019 AIAA Modeling and Simulation Technologies Conference.
  San Diego, CA. Jan}, volume 2109, 2019.

\bibitem{weiss2006global}
Bernhard Wei{\ss}, Michael Naderhirn, and Luigi del Re.
\newblock Global real-time path planning for uavs in uncertain environment.
\newblock In {\em 2006 IEEE Conference on Computer Aided Control System Design,
  2006 IEEE International Conference on Control Applications, 2006 IEEE
  International Symposium on Intelligent Control}, pages 2725--2730. IEEE,
  2006.

\bibitem{NASAwhite}
National Aeronautics and Space Administration.
\newblock Unmanned aircraft system (uas) traffic management (utm) enabling
  civilian low-altitude airspace and unmanned aircraft system operations, 2020.

\bibitem{lopez1995autonomous}
Ismael Lopez and Colin~R Mclnnes.
\newblock Autonomous rendezvous using artificial potential function guidance.
\newblock {\em Journal of Guidance, Control, and Dynamics}, 18(2):237--241,
  1995.

\bibitem{ruchti2014unmanned}
Jason Ruchti, Robert Senkbeil, James Carroll, Jared Dickinson, James Holt, and
  Saad Biaz.
\newblock Unmanned aerial system collision avoidance using artificial potential
  fields.
\newblock {\em Journal of Aerospace Information Systems}, 11(3):140--144, 2014.

\bibitem{arcak2007passivity}
Murat Arcak.
\newblock Passivity as a design tool for group coordination.
\newblock {\em IEEE Transactions on Automatic Control}, 52(8):1380--1390, 2007.

\bibitem{reynolds1987flocks}
Craig~W Reynolds.
\newblock Flocks, herds and schools: A distributed behavioral model.
\newblock In {\em Proceedings of the 14th annual conference on Computer
  graphics and interactive techniques}, pages 25--34, 1987.

\bibitem{olfati2006flocking}
Reza Olfati-Saber.
\newblock Flocking for multi-agent dynamic systems: Algorithms and theory.
\newblock {\em IEEE Transactions on automatic control}, 51(3):401--420, 2006.

\bibitem{tanner2007flocking}
Herbert~G Tanner, Ali Jadbabaie, and George~J Pappas.
\newblock Flocking in fixed and switching networks.
\newblock {\em IEEE Transactions on Automatic control}, 52(5):863--868, 2007.

\bibitem{lei2008flocking}
Bin Lei, Wenfeng Li, and Fan Zhang.
\newblock Flocking algorithm for multi-robots formation control with a target
  steering agent.
\newblock In {\em 2008 IEEE International Conference on Systems, Man and
  Cybernetics}, pages 3536--3541. IEEE, 2008.

\bibitem{li2009flocking}
Qin Li and Zhong-Ping Jiang.
\newblock Flocking control of multi-agent systems with application to
  nonholonomic multi-robots.
\newblock {\em Kybernetika}, 45(1):84--100, 2009.

\bibitem{gouvea2011potential}
Josiel~A Gouvea, Fernando Lizarralde, and Liu Hsu.
\newblock Potential function formation control of nonholonomic mobile robots
  with curvature constraints.
\newblock {\em IFAC Proceedings Volumes}, 44(1):11931--11936, 2011.

\bibitem{jin2015consensus}
Jingfu Jin, Yoon-Gu Kim, Sung-Gil Wee, and Nicholas Gans.
\newblock Consensus based attractive vector approach for formation control of
  nonholonomic mobile robots.
\newblock In {\em 2015 IEEE International Conference on Advanced Intelligent
  Mechatronics (AIM)}, pages 977--983. IEEE, 2015.

\bibitem{mastellone2008formation}
Silvia Mastellone, Du{\v{s}}an~M Stipanovi{\'c}, Christopher~R Graunke, Koji~A
  Intlekofer, and Mark~W Spong.
\newblock Formation control and collision avoidance for multi-agent
  non-holonomic systems: Theory and experiments.
\newblock {\em The International Journal of Robotics Research}, 27(1):107--126,
  2008.

\bibitem{sun2019hybrid}
Dawei Sun, Cheolhyeon Kwon, and Inseok Hwang.
\newblock Hybrid flocking control algorithm for fixed-wing aircraft.
\newblock {\em Journal of Guidance, Control, and Dynamics}, 42(11):2443--2455,
  2019.

\bibitem{kuriki2014consensus}
Yasuhiro Kuriki and Toru Namerikawa.
\newblock Consensus-based cooperative formation control with collision
  avoidance for a multi-uav system.
\newblock In {\em 2014 American Control Conference}, pages 2077--2082. IEEE,
  2014.

\bibitem{van2007modeling}
Tom Van~Woensel and Nico Vandaele.
\newblock Modeling traffic flows with queueing models: a review.
\newblock {\em Asia-Pacific Journal of Operational Research}, 24(04):435--461,
  2007.

\bibitem{tanaka1991homoclinic}
Kazunaga Tanaka.
\newblock Homoclinic orbits in a first order superquadratic hamiltonian system:
  convergence of subharmonic orbits.
\newblock {\em Journal of Differential Equations}, 94(2):315--339, 1991.

\bibitem{kozlov2003weak}
Valery~Vasil'evich Kozlov and DV~Treshchev.
\newblock Weak convergence of solutions of the liouville equation for nonlinear
  hamiltonian systems.
\newblock {\em Theoretical and mathematical physics}, 134(3):339--350, 2003.

\bibitem{broyden1969new}
CG~Broyden.
\newblock A new method of solving nonlinear simultaneous equations.
\newblock {\em The Computer Journal}, 12(1):94--99, 1969.

\bibitem{helmke2012optimization}
Uwe Helmke and John~B Moore.
\newblock {\em Optimization and dynamical systems}.
\newblock Springer Science \& Business Media, 2012.

\bibitem{alvarez2000minimizing}
Felipe Alvarez.
\newblock On the minimizing property of a second order dissipative system in
  hilbert spaces.
\newblock {\em SIAM Journal on Control and Optimization}, 38(4):1102--1119,
  2000.

\bibitem{rigatos2012nonlinear}
Gerasimos~G Rigatos.
\newblock Nonlinear kalman filters and particle filters for integrated
  navigation of unmanned aerial vehicles.
\newblock {\em Robotics and Autonomous Systems}, 60(7):978--995, 2012.

\bibitem{bibuli2009path}
Marco Bibuli, Gabriele Bruzzone, Massimo Caccia, and Lionel Lapierre.
\newblock Path-following algorithms and experiments for an unmanned surface
  vehicle.
\newblock {\em Journal of Field Robotics}, 26(8):669--688, 2009.

\bibitem{yu2015sense}
Xiang Yu and Youmin Zhang.
\newblock Sense and avoid technologies with applications to unmanned aircraft
  systems: Review and prospects.
\newblock {\em Progress in Aerospace Sciences}, 74:152--166, 2015.

\end{thebibliography}
\end{document}